\documentclass[conference]{IEEEtran}
\usepackage{diagbox} 
\usepackage{graphicx}
\usepackage{epstopdf}
\usepackage{psfrag}
\usepackage{subfigure}
\usepackage{url}
\usepackage{stfloats}
\usepackage{amsfonts,amssymb,amsmath,bm,paralist,cite,ifthen,color}
\usepackage{array}
\usepackage{caption}
\usepackage{calc}
\usepackage{enumerate}
\usepackage{algorithmic,algorithm}
\usepackage{array}
\usepackage{float}
\usepackage[colorlinks,
            linkcolor=black,
            anchorcolor=black,
            citecolor=black]{hyperref}

\newtheorem{theorem}{Theorem}

\newtheorem{proof}{Proof}

\newcommand{\qedsymbol}{\hfill \(\blacksquare\)}
\graphicspath{{Figures/}}

\newcommand\Cc{\ensuremath{\mathcal{C}}}
\newcommand\Nc{\ensuremath{\mathcal{N}}}
\newcommand\Oc{\ensuremath{\mathcal{O}}}

\newcommand\xb{\ensuremath{{\bm x}}}
\newcommand\ab{\ensuremath{{\bm a}}}

\newcommand\nb{\ensuremath{{\bf n}}}

\newcommand\yb{\ensuremath{{\bm y}}}
\newcommand\ssb{\ensuremath{{\bm s}}}

\newcommand\Hb{\ensuremath{{\bm H}}}

\newcommand\Bb{\ensuremath{{\bm B}}}
\newcommand\Ab{\ensuremath{{\bm A}}}
\newcommand\Cb{\ensuremath{{\bm C}}}

\newcommand\Db{\ensuremath{{\bm D}}}

\newcommand\Fb{\ensuremath{{\bm F}}}

\newcommand\Ib{\ensuremath{{\bm I}}}

\newcommand\Tb{\ensuremath{{\bm T}}}

\newcommand\Ub{\ensuremath{{\bm U}}}

\newcommand\Qb{\ensuremath{{\bm Q}}}
\newcommand\Vb{\ensuremath{{\bm V}}}
\newcommand\vb{\ensuremath{{\bm v}}}
\newcommand\Wb{\ensuremath{{\bm W}}}
\newcommand\Jb{\ensuremath{{\bm J}}}

\newcommand\Gammab{\ensuremath{{\bm \Gamma}}}
\newcommand\Lambdab{\ensuremath{{\bm \Lambda}}}

\newcommand\Cs{\ensuremath{{\mathbb{C}}}}
\newcommand\Es{\ensuremath{{\mathbb{E}}}}

\IEEEoverridecommandlockouts

\begin{document}

\title{ Closed-Form Hybrid Beamforming Solution for Spectral Efficiency Upper Bound Maximization in mmWave MIMO-OFDM Systems 
\thanks{This work has been supported in part by Academy of Finland under 6Genesis Flagship (grant 318927) and EERA Project (grant 332362).}
}

\author{\IEEEauthorblockN{Mengyuan Ma, Nhan Thanh Nguyen and Markku Juntti}
\IEEEauthorblockA{\textit{Centre for Wireless Communications (CWC), Uninvesity of Oulu, P.O.Box 4500, FI-90014, Finland}\\
Email: \{mengyuan.ma, nhan.nguyen, markku.juntti\}@oulu.fi}
}
%

\maketitle

\begin{abstract}
Hybrid beamforming is considered a key enabler to realize millimeter wave (mmWave) multiple-input multiple-output (MIMO) communications due to its capability of considerably reducing the number of costly and power-hungry radio frequency chains in the transceiver. However, in mmWave MIMO orthogonal frequency-division multiplexing (MIMO-OFDM) systems, hybrid beamforming design is challenging because the analog precoder and combiner are required to be shared across the whole employed bandwidth. In this paper, we propose closed-form solutions to the problem of designing the analog precoder/combiner in a mmWave MIMO-OFDM system by maximizing the upper bound of the spectral efficiency. The closed-form solutions facilitate the design of analog beamformers while guaranteeing state-of-art performance. Numerical results show that the proposed algorithm attains a slightly improved performance with much lower computational complexity compared to the considered benchmarks.
\end{abstract}


%
\IEEEpeerreviewmaketitle

\section{Introduction}
Millimeter wave (mmWave) communication system, with its large utilizable spectrum, is promising to meet the increasing demand for the data rate of current and future wireless networks \cite{swindlehurst2014millimeter,aldubaikhy2020mmwave}. Hybrid beamforming (HBF) for mmWave multiple-input multiple-output (MIMO) communications has attracted great interest in recent years since it guarantees the spatial multiplexing gain with the requirement of reduced power-hungry radio frequency (RF) chains \cite{ahmed2018survey}. 

 However, the design of  HBF in orthogonal frequency-division multiplexing (OFDM) is considerably challenging due to the analog beamformer is required to be shared among the whole bandwidth. This motivates various HBF designs for mmWave MIMO-OFDM systems \cite{yu2016alternating,sohrabi2017hybrid,tsai2018sub,zilli2021constrained,ku2021low}. Specifically, Yu \textit{et al.} in \cite{yu2016alternating} propose a beamforming design based on minimizing the total norm distances between the HBF and optimal unconstrained beamforming matrices for all the subcarriers. To solve the problem, they propose an iterative algorithm by an alternating optimization approach. In \cite{sohrabi2017hybrid}, the analog beamformers are designed by maximizing the upper bound of spectral efficiency (SE). The solution is obtained by an iterative optimization method. In \cite{tsai2018sub}, Tsai \textit{et al.} propose a low-complexity procedure in which the analog precoder and combiner are designed by respectively extracting the phases of the averaged channel covariance matrix and the averaged conjugate-transposed channel covariance matrix. In \cite{zilli2021constrained}, Zilli \textit{et al.} transform the analog beamformer design into a Tucker2 tensor decomposition problem, allowing to find the solution by the projected alternate least-square-based algorithm. In \cite{ku2021low}, Ku \textit{et al.} propose a reduced-complexity algorithm in which the sub-optimization problems of \cite{sohrabi2017hybrid} are iteratively solved to bypass the matrix-inversion operation.

 In most of the aforementioned works, the analog beamformers are obtained by iterative optimization approaches. Such designs entail excessively high computational burdens of the transceiver, making HBF less efficient, especially in massive MIMO systems \cite{zhang2019hybrid}. To overcome this challenge, in this work, we propose a novel scheme for HBF design and optimization. Specifically, we focus on the problem of maximizing the upper bound of the system SE. The closed-form solutions to the analog beamformers are found efficiently by partial singular value decomposition (SVD). The proposed scheme is computationally efficient and feasible for practical implementation because it does not perform any iterative procedure. Nevertheless, it still achieves slight performance improvement compared with state-of-art HBF algorithms.


\section{MmWave MIMO-OFDM System}


\subsection{Signal Model}
   \begin{figure}[htbp]
        \centering	
         \includegraphics[width=0.5\textwidth]{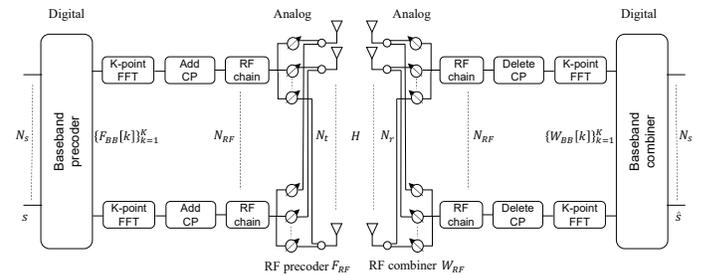}
         \captionsetup{font={small}}
        \caption{HBF transceiver structure in a mmWave MIMO-OFDM system.}

        	\label{fig:HBF transceiver structure in mmWave MIMO-OFDM system}
        	\vspace{-0.3cm}
  \end{figure}
We consider a MIMO-OFDM system where the transmitter and receiver are equipped with $N_t$ and $N_r$ antennas, respectively. Without loss of generality, we assume that transmitter and receiver are deployed with the same number of RF chains, denoted by $N_{RF}$ ($N_{RF}\leq \min(N_t,N_r)$). Let $K$ be the number of subcarriers, and $\ssb_k \in \Cs^{N_s \times 1}$ $(N_s\leq N_{RF})$ be the transmitted symbol vector at the $k$th subcarrier, $\Es\left[\ssb_k\ssb^H_k\right]=\Ib_{N_s},\; k=1,\cdots,K$, where $\Ib_{N_s}$ denotes the $N_s\times N_s$ identity matrix. As illustrated in Fig. \ref{fig:HBF transceiver structure in mmWave MIMO-OFDM system}, the transmitted vector $\ssb_k$ is first precoded by the low-dimensional baseband precoder $\Fb_{BB}[k] \in \Cs ^{N_{RF} \times N_s}$ and then transformed into time domain by  $N_{RF}$ $K$-point inverse fast Fourier transforms (IFFTs). The signal vector in time domain is then added the cyclic prefix (CP) and precoded by the common high-dimensional analog precoder, denoted by $\Fb_{RF}\in \Cs ^{N_t \times N_{RF}}$ satisfying $|\Fb_{RF}(i,j)|=1/\sqrt{N_t}, \forall i,j$, for all subcarriers. Here, $\Ab(i,j)$ denotes the entry on the $i$th row and $j$th colum of matrix  $\Ab$, and $|x|$ represents the modulus of the complex number $x$. The final transmitted signal for each subcarrier is given as
\begin{equation}\label{eq:transmitted signal}
  \xb_k=\Fb_k\ssb_k
\end{equation}
where $\Fb_k=\Fb_{RF}\Fb_{BB}[k]$ and $\| \Fb_k \|^2_F \leq P_b$, with $P_b$ being the power budget for each subcarrier.

At the receiver, the received signal vector is first combined by the common analog combiner, represented by $\Wb_{RF}\in \Cs ^{N_r \times N_{RF}}$,  $|\Wb_{RF}(i,j)|=1/\sqrt{N_r}, \forall i,j$. After discarding the CP and performing $N_{RF}$ $K$-point fast Fourier transforms (FFTs), the combined signal is further processed at frequency domain by low-dimensional baseband combiner $\Wb_{BB}[k] \in \Cs ^{N_{RF} \times N_s}$ for each subcarrier. Finally, the received signal at the $k$th subcarrier through channel $\Hb_k\in \Cs^{N_r\times N_t}$ is given as
\begin{equation}\label{eq:received signal}
  \yb_k=\Wb^H_k\Hb_k \Fb_k\ssb_k + \Wb^H_k\nb_k
\end{equation}
where $\Wb_k=\Wb_{RF}\Wb_{BB}[k]$, and $\nb_k\sim \Nc(\boldsymbol{0}, \sigma_n^2\Ib_{N_r})$ is the additive white Gaussian noise vector at the $k$th subcarrier with $\sigma_n^2$ being the noise variance.

\subsection{Channel Model}\label{sc:channel model}
The mmWave channel can be characterized by the Saleh-Valenzuela model \cite{rappaport2015millimeter}. For OFDM systems, the frequency-selective channel is divided into $K$ frequency-flat sub-channels with non-interfering subcarriers. In frequency domain, the channel matrix at the $k$th subcarrier can be given by \cite{sohrabi2017hybrid,zilli2021constrained}
\begin{equation}\label{eq:channel model}
  \Hb_k=\sqrt{\frac{N_rN_t}{N_{cl} N_{ray}}}\sum_{c=0}^{N_{cl}-1}\sum_{\ell=1}^{N_{ray}}\alpha_{c,\ell} \ab_r(\theta^r_{c,\ell})\ab_t^H(\theta^t_{c,\ell})e^{\frac{-j2\pi c(k-1)}{K}},
\end{equation}
where $\alpha_{c,\ell} \sim \Cc\Nc(0,1)$, $\theta^r_{c,\ell}$, and $\theta^t_{c,\ell}$ are the path gain, angles of arrival and departure of the $\ell$th path in the $c$th cluster, respectively. Furthermore, $\theta^r_{c,\ell}$ and $\theta^t_{c,\ell}$ follow a Laplacian distribution with mean angles $\theta^r_{c}$, $\theta^t_{c}$ uniformly distributed over $[0,2\pi)$, and angular spread of $\sigma_r$ and $\sigma_t$ \cite{el2014spatially}. $\ab_r(\cdot)$ and $\ab_t(\cdot)$ are array response vectors at the receiver and transmitter, respectively. For simplicity, we assume that the numbers of scattering clusters and propagation paths, and the angular spreads are the same for all subcarriers. Furthermore, we assume the uniform linear array (ULA) is employed. The steering vector of ULA with $N$ antennas and spacing distance $d$ of adjacent antenna is modeled as \cite{sayeed2002deconstructing}
\begin{equation}\label{eq:steering vector}
  \ab(\theta)=\frac{1}{\sqrt{N}}\left[1,e^{-j\frac{2\pi}{\lambda}d\sin(\theta)},\cdots,e^{-j\frac{2\pi}{\lambda}(N-1)d\sin(\theta)}\right]^T,
\end{equation}
where $\lambda$ is the signal wavelength. Note that in this representation, the operating frequency is assume to be much larger than the total bandwidth of the OFDM system such that signal wavelength in all subcarriers is approximately the same.

\subsection{Problem Formulation}
 In this work, we focus on the problem of HBF design with the assumption that the full channel state information (CSI) is available at both the transmitter and receiver \cite{yu2016alternating,sohrabi2017hybrid,zilli2021constrained,ku2021low}. We aim at designing transmit and receive HBF matrices that maximize the overall SE of the system under power constraints for subcarriers and constant-modulus constraints for analog beamformers. The problem is expressed as
\begin{subequations}\label{eq:problem formulation}
  \begin{align}
    \underset{\Fb_{RF},\Wb_{RF},\atop \left\{ \Fb_{BB}[k], \Wb_{BB}[k]\right\}_{k=1}^{K}}{\max} \quad  & \frac{1}{K} \sum_{k=1}^{K} R_k\\
    \textrm{s.t.} \qquad  & \| \Fb[k] \|^2_F \leq P_b ,\\
    &\left|\Fb_{RF}(i,j)\right|=1/\sqrt{N_t}, \quad \forall i, j, \label{eq: analog precoder constraint}\\
    &\left|\Wb_{RF}(i,j)\right|=1/\sqrt{N_r}, \quad \forall i, j, \label{eq: analog combiner constraint}
\end{align}
\end{subequations}
where $R_k$ denotes the SE of the $k$th subcarrier. Considering the Gaussian signalling, $R_k$ is expressed as
\begin{equation}\label{eq:SE for each subcarrier}
  R_k=\log_{2}\left|\Ib_{N_{s}}+\frac{1}{\sigma_n^{2}} \Wb_k^{\dagger}  \Hb_k \Fb_k \Fb^H_k \Hb^H_k \Wb_k \right|,
\end{equation}
where $\dagger$ denotes the Moore-Penrose pseudo inversion, $\Fb_k=\Fb_{RF}\Fb_{BB}[k]$ and $\Wb_k=\Wb_{RF}\Wb_{BB}[k]$.

\section{Hybrid beamforming design}
The SE maximization problem in (\ref{eq:problem formulation}) requires a joint optimization over both the transmitter and receiver beamformers, i.e., $\left\{ \Fb_k \right\}_{k=1}^{K}$ and $\left\{ \Wb_k \right\}_{k=1}^{K}$. It is non-convex due to the unit-modulus constraints \eqref{eq: analog precoder constraint} and \eqref{eq: analog combiner constraint} of the analog beamformers. Therefore, it is challenging to obtain the optimal solution \cite{sohrabi2017hybrid}. To tackle this challenge, we decouple the problem into two subproblems. Specifically, for the design of transmitter beamformers $\left\{ \Fb_k \right\}_{k=1}^{K}$, we assume that the optimal receiver is used. Then, the receiver beamformers $\left\{ \Wb_k \right\}_{k=1}^{K}$ are obtained given the transmitter beamformers $\left\{ \Fb_k \right\}_{k=1}^{K}$. The solutions to these subproblems are presented in the following subsections.

\subsection{Transmitter HBF Design}
Given the receive beamforming matrices $\left\{ \Wb_k \right\}_{k=1}^{K}$, the hybrid beamforming design problem at the transmitter is expressed as
\begin{subequations}\label{eq:problem formulation at transmitter}
  \begin{align}
        \max\limits_{\Fb_{RF}, \left\{ \Fb_{BB}[k]\right\}_{k=1}^{K} } & \frac{1}{K} \sum_{k=1}^{K} \tilde{R}_k\\
        \text { s.t. } \qquad & \| \Fb_{RF}\Fb_{BB}[k] \|^2_F \leq P_b, \\
        &\left|\Fb_{RF}(i,j)\right|=1/\sqrt{N_t}, \quad \forall i, j, 
    \end{align}
\end{subequations}
where $\tilde{R}_k$ is given as
\begin{equation}\label{eq:SE for transmitter}
  \tilde{R}_k=\log_{2}\left|\Ib+\frac{1}{\sigma_n^{2}} \Hb_k \Fb_{RF}\Fb_{BB}[k] \Fb_{BB}^H[k]\Fb_{RF}^H \Hb^H_k  \right|.
\end{equation}
When $\Fb_{RF}$ is fixed, the optimal digital precoder is given as
\begin{equation}\label{eq:digital precoder}
  \Fb_{BB}[k]=(\Fb_{RF}^H\Fb_{RF})^{-\frac{1}{2}}\Vb_k \Gammab^{\frac{1}{2}}_k,
\end{equation}
where $\Vb_k \in \Cs^{N_{RF}\times N_{s}}$ is the set of right singular vectors corresponding to the $N_s$ largest singular values of $\Qb_k \triangleq \Hb_k\Fb_{RF}(\Fb_{RF}^H\Fb_{RF})^{-\frac{1}{2}}$, and $\Gammab_k$ is the diagonal matrix with its diagonal elements being the allocated power to each symbol of the subcarrier $k$.

Since it is challenging to design the analog beamforming matrix $\Fb_{RF}$ that is shared among all the subcarriers, we alternatively tend to find the solution to the upper bound of the objective function in problem (\ref{eq:problem formulation at transmitter}). For moderate and high signal-to-noise-ratio (SNR) regime, it is asymptotically optimal to adopt the equal power allocation for all streams in each subcarrier \cite{lee2008downlink}. In this regime, we have $\Gammab_k=\gamma\Ib_{N_s}$ and $\gamma=P_b/N_s$. Substituting (\ref{eq:digital precoder}) into (\ref{eq:SE for transmitter}), we obtain
\begin{align}\label{eq: transmitter SE relaxtion}
  \tilde{R}_k \overset{a}= & \log_{2}\left|\Ib+\frac{\gamma}{\sigma_n^2} \tilde{\Vb}_k^H\Qb_k^H \Qb_k \tilde{\Vb}_k \tilde{\Ib}_{N_{RF}}  \right| \notag \\
  \overset{b}\leq & \log_{2}\left| \Ib + \frac{\gamma}{\sigma_n^2} \tilde{\Vb}_k^H\Qb_k^H \Qb_k \tilde{\Vb}_k \Ib_{N_{RF}} \right| \notag\\
  =& \log_{2}\left| \Ib + \frac{\gamma}{\sigma_n^2}\Qb_k^H \Qb_k  \right|,
\end{align}
where the equality in (a) is obtained based on
\begin{equation}\label{eq:unitary matrix equality}
  \Vb_k\Vb_k^H=\tilde{\Vb}_k \tilde{\Ib}_{N_{RF}} \tilde{\Vb}_k ^H
\end{equation}
with $\tilde{\Vb}_k \in \Cs^{N_{RF}\times N_{RF}}$ being a unitary matrix and  $\tilde{\Ib}_{N_{RF}}=\left[\begin{array}{rr}
\Ib_{N_{s}} & \mathbf{0} \\
\mathbf{0} & \mathbf{0}
\end{array}\right]$ being a $N_{RF}\times N_{RF}$ matrix. Furthermore, the equation in (b) is because $\tilde{\Vb}_k^H\Qb_k^H \Qb_k \tilde{\Vb}_k $ is a semi-definite positive matrix. Note that the relaxation of the inequality (b) is tight when $N_s=N_{RF}$ \cite{sohrabi2017hybrid}.

Since $\Fb_{RF}^H \Fb_{RF} \rightarrow \Ib_{N_{RF}}$ when the array size is very large \cite{el2014spatially}, it is a reasonable assumption that the analog precoder $\Fb_{RF}$ is a matrix of full-column rank. In this case, the independence between different analog beamformers for different RF chains endows the largest optimization space for finding the optimal digital beamformer. Based on the QR decomposition, we can express $\Fb_{RF}=\Vb_{RF}\Bb$, where $\Vb_{RF} \in \Cs^{N_{t}\times N_{RF}}$ has orthogonal column vectors, i.e., $\Vb_{RF}^H\Vb_{RF}=\Ib_{N_{RF}}$, and $\Bb$ is an invertible matrix. The upper bound of the objective function in (\ref{eq:problem formulation at transmitter}) is
\begin{align}\label{eq: transmitter SE upper bound}
  \frac{1}{K} \sum_{k=1}^{K}\tilde{R}_k \leq &\frac{1}{K} \sum_{k=1}^{K} \log_{2}\left| \Ib + \frac{\gamma}{\sigma_n^2}\Qb_k^H \Qb_k \right| \notag \\
  = & \frac{1}{K} \sum_{k=1}^{K} \log_{2}\left| \Ib + \frac{\gamma}{\sigma_n^2} \Vb_{RF}^H \Hb_k^H \Hb_k \Vb_{RF}\right| \notag\\
   \overset{c}\leq& \log_{2}\left| \Ib + \frac{\gamma}{\sigma_n^2}\Vb_{RF}^H \Hb_e \Vb_{RF} \right|,
\end{align}
where $\Hb_e\triangleq \frac{1}{K} \sum_{k=1}^{K}\Hb_k^H \Hb_k$, and the equation in (c) is based on Jensen's inequality. In summary, the analog precoder $\Fb_{RF}$ can be obtained by solving the following problem
\begin{align}\label{eq:analog precoder problem}
  \max\limits_{\Fb_{RF}} &\quad \log_{2}\left| \Ib + \frac{\gamma}{\sigma_n^2}\Vb_{RF}^H \Hb_e \Vb_{RF} \right| \\
   \text { s.t. } & \left|\Fb_{RF}(i,j)\right|=1/\sqrt{N_t}, \quad \forall i, j. \notag
\end{align}
The unconstrained solution to (\ref{eq:analog precoder problem}), denoted by $\Fb_{RF}^*$, is given in the following theorem.
 \begin{theorem}\label{tr:optimal solution}
       Let $\Hb_e=\Vb_e\Lambdab_e\Vb_e^H$, where $\Vb_e=[\vb_{e,1},\cdots,\vb_{e,N_t}]$ is a unitary matrix, and $\Lambdab_e$ is a diagonal matrix with diagonal elements being the eigenvalues of $\Hb_e$ in a decreasing order. The unconstrained solution to (\ref{eq:analog precoder problem}) is given as
       \begin{equation}\label{eq:optiaml solution without hardware constraints}
         \Fb_{RF}^*=\Vb_{RF}^*\Cb,
       \end{equation}
       where $\Vb_{RF}^*=[\vb_{e,1},\cdots,\vb_{e,N_{RF}}]$ and $\Cb\in \Cs^{N_{RF}\times N_{RF}}$ is an arbitrary invertible matrix.
 \end{theorem}
\begin{proof}
  See in Appendix \ref{sec_proof}. \qedsymbol 
\end{proof}

As a result, the constrained solution is obtained by projecting $\Fb_{RF}^*$ into the feasible space \cite{zhang2005variable}, i.e.,
\begin{equation}\label{eq:optimal solution to analog precoder}
  \Fb_{RF}(i,j)=\frac{1}{\sqrt{N_t}}e^{j\angle\Fb_{RF}^*(i,j)},\; \forall i,j,
\end{equation}
where $\angle x$ denotes the phase of the complex number $x$.
\subsection{Receiver HBF Design}
For a fixed analog combiner $\Wb_{RF}$, the optimal digital combiner of each subcarrier is the MMSE solution, i.e.,
 \begin{equation}\label{eq:optimal baseband combiner}
   \Wb_{BB}^*[k]=\left(\Jb_k\Jb^H_k+\sigma_n^2\Wb_{RF}^H\Wb_{RF}\right)^{-1}\Jb_k,
 \end{equation}
where $\Jb_k\triangleq \Wb_{RF}^H \Hb_k\Fb_k$. Then, the problem of analog combiner design can be formulated as
 \begin{align}\label{eq:analog combiner problem}
    \max\limits_{\Wb_{RF}} & \; \frac{1}{K}\sum_{k=1}^{K}\log_2\left|\Ib+\frac{1}{\sigma_n^2}\Wb_{RF}^{\dagger}\Tb_k\Wb_{RF}\right| \\
   \text { s.t. } & \left|\Wb_{RF}(i,j)\right|=1/\sqrt{N_r}, \quad \forall i, j, \notag
 \end{align}
 where $\Tb_k\triangleq \Hb_k\Fb_k\Fb^H_k\Hb^H_k$. Based on the QR decomposition, we can express $\Wb_{RF}=\Ub_{RF}\Db$, where $\Ub_{RF} \in \Cs^{N_{r}\times N_{RF}}$ has orthogonal column vectors, i.e., $\Ub_{RF}^H\Ub_{RF}=\Ib_{N_{RF}}$, and $\Db$ is an invertible matrix. By applying the Jensen's inequality, we can consider maximizing the upper bound of (\ref{eq:analog combiner problem}) for designing $\Wb_{RF}$ as,
 \begin{align}\label{eq:analog cpmbiner problem upper bound}
  \max\limits_{\Wb_{RF}} &\quad \log_{2}\left| \Ib + \frac{1}{\sigma_n^2}\Ub_{RF}^H \Tb_e \Ub_{RF} \right| \\
   \text { s.t. } & \left|\Wb_{RF}(i,j)\right|=1/\sqrt{N_r}, \quad \forall i, j, \notag
\end{align}
where $\Tb_e\triangleq\frac{1}{K} \sum_{k=1}^{K} \Tb_k$. The optimal solution to problem (\ref{eq:analog cpmbiner problem upper bound}) can be obtained similarly as that to (\ref{eq:analog precoder problem}). 

Finally, based on \eqref{eq:digital precoder}, \eqref{eq:optiaml solution without hardware constraints},\eqref{eq:optimal solution to analog precoder}, \eqref{eq:optimal baseband combiner}, and \eqref{eq:analog cpmbiner problem upper bound}, the proposed HBF scheme is summarized in Algorithm 1.

\begin{table}[htbp]
  \vspace{-0.5pt}
  \centering
  \begin{tabular}{cl} 
    \hline
    &{{\bf Algorithm 1}: The proposed HBF scheme for mmWave}\\
    &\qquad \qquad \quad \; {  MIMO-OFDM systems}  \\
    \hline
    &1. Input: $\Hb_k, \gamma, \sigma^2_n$\\
    &2. \quad $\Hb_e = \frac{1}{K} \sum_{k=1}^{K}\Hb_k^H \Hb_k$.\\
    &3. \quad Find $\Fb_{RF}$ according to \eqref{eq:optimal solution to analog precoder}. \\
    &4. \quad Calculate $\Fb_{BB}[k]$ according to \eqref{eq:digital precoder}.\\
    &5. \quad$\Tb_e=\frac{1}{K} \sum_{k=1}^{K}\Hb_k\Fb_{RF}\Fb_{BB}[k] (\Hb_k\Fb_{RF}\Fb_{BB}[k])^H$.\\
    &6. \quad Find $\Wb_{RF}$ by solving problem \eqref{eq:analog cpmbiner problem upper bound}.\\
    &7. \quad Calculate $\Wb_{BB}[k]$ according to \eqref{eq:optimal baseband combiner}.\\
    &8. Output: $\Fb_{RF}$, $\Wb_{RF}$, $\Fb_{BB}[k]$, $\Wb_{BB}[k]$.  \\
    \hline
  \end{tabular}
\end{table}

\section{Results}
In this section, we present the numerical results to evaluate the performance and computational complexity of the proposed scheme. In the simulations, we use the channel model given in Section \ref{sc:channel model} with $N_{cl}=5$, $N_{ray}=10$, $\sigma_r=\sigma_t=10^\circ$, and antenna spacing distance $d=\frac{\lambda}{2}$. The SNR is defined as SNR$\triangleq P_b/\sigma_n^2$. All reported results are averaged over 100 channel realizations. In the simulations, we consider the hybrid beamforming in large-scale antenna arrays (HBF-LSAA) algorithm in \cite{sohrabi2017hybrid} and the low-complexity version of HBF-LSAA (HBF-LSAA-fast) algorithm in \cite{ku2021low} for comparison. The performance of optimal digital beamforming (DBF) via water-filling algorithm is also presented.
   \begin{figure}[htbp]
   \vspace{-0.3cm}
        \centering	
         \includegraphics[width=0.4\textwidth]{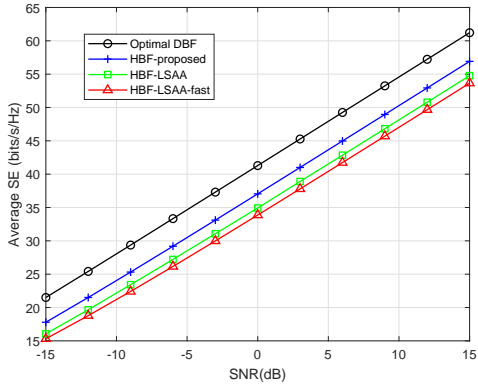}
         \captionsetup{font={small}}
        \caption{SE of considered HBF algorithms vs SNR with $N_t=N_r=64,N_s=N_{RF}=4,K=512$.}

        	\label{fig:SE of different vs SNR}
        	\vspace{-0.3cm}
  \end{figure}
   \begin{figure}[htbp]
   \vspace{-0.3cm}
        \centering	
         \includegraphics[width=0.4\textwidth]{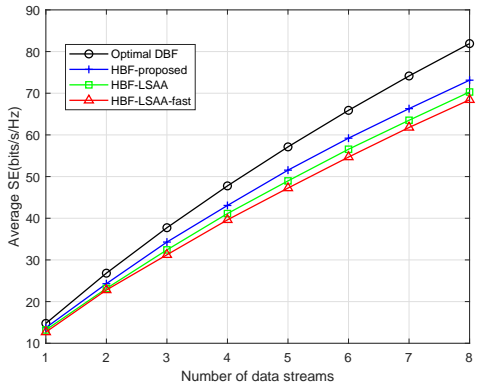}
         \captionsetup{font={small}}
        \caption{SE of considered HBF algorithms vs $N_s$ with $N_t=N_r=64,N_s=N_{RF},K=512$, SNR $= 5$ dB.}

        	\label{fig:SE of different HBF design vs Ns}
      \vspace{-0.3cm}
  \end{figure}
  
In Fig. \ref{fig:SE of different vs SNR}, we show the average SE of considered HBF algorithms versus the SNR of a OFDM system with $N_t=N_r=64,N_{RF}=N_s=4,K=512$. It is observed that the proposed method outperforms the compared algorithms over the entire considered SNR range. For the same SE, the SNR of the proposed design is required to be, on average, $1.5$dB lower than that of HBF-LSAA and $2.5$dB lower than that of HBF-LSAA-fast.

 In Fig. \ref{fig:SE of different HBF design vs Ns}, we show the average SE of the proposed HBF algorithm, the HBF-LSAA and HBF-LSAA-fast algorithms with respect to the number of data stream $N_s$ of a OFDM system with $N_t=N_r=64,K=512,N_{RF}=N_s$, and SNR$=5$dB. It is observed from this figure that the proposed HBF algorithm always outperforms the HBF-LSAA and HBF-LSAA-fast algorithms for the considered number of data streams.  

Next, we evaluate the computational complexity of the proposed HBF algorithm, the HBF-LSAA and HBF-LSAA-fast algorithms, which is computed as the number of floating-point operations (FLOPs). Because the main difference of the considered algorithms lies in the design of analog beamformers, we only show the complexities required to obtain analog beamformers. The computational complexity of the three algorithms are summarized in Table \ref{Tb:complexity}, where $N=\{N_r,N_t\}$. For the HBF-LSAA-fast algorithm, $N_{iter}$ is the number of iteration required to obtain the first dominant vector by the power method \cite{bentbib2015block}. Note that the computational complexity of the proposed method is mainly caused by performing the SVD. Since we only take the first $N_{RF}$ dominant singular vectors, the SVD here is actually the partial SVD which can be obtained at the complexity $\Oc(N_{RF}N^2)$ \cite{halko2011finding}. It is observed from Table \ref{Tb:complexity} that the proposed HBF algorithm is more computationally efficient compared with the HBF-LSAA and HBF-LSAA-fast algorithms. This is because the proposed HBF algorithm does not require iteration to obtain the analog beamformers.
\begin{table}[htbp]
    \centering
    \caption{Computational complexity of considered algorithms}\label{Tb:complexity}
        \begin{tabular}{|c|c|}
            \hline
            Algorithm &  Computational complexity\\
            \hline
            HBF-LSAA & $\Oc(N_{RF}N^3+N_{RF}^2N^2+N_{RF}^4)$\\
            \hline
            HBF-LSAA-fast & $\Oc(N_{iter}N_{RF}N^2)$\\
            \hline
            HBF-proposed & $\Oc(N_{RF}N^2)$ \\
            \hline
        \end{tabular}
        \vspace{-0.3cm}
\end{table}
     \begin{figure}[htbp]
     \vspace{-0.3cm}
        \centering	
         \includegraphics[width=0.4\textwidth]{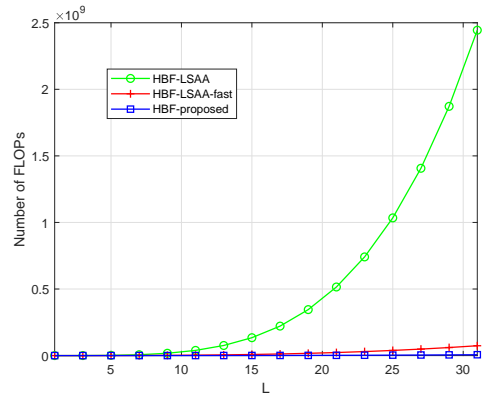}
         \captionsetup{font={small}}
        \caption{Computational complexity vs $L$ with $(N_r,Ns)=L*(8,1),N_r=N_t,N_s=N_{RF}$.}

        	\label{fig:Computational complexity}
        	\vspace{-0.3cm}
  \end{figure}
       \begin{figure}[htbp]
       \vspace{-0.3cm}
        \centering	
         \includegraphics[width=0.4\textwidth]{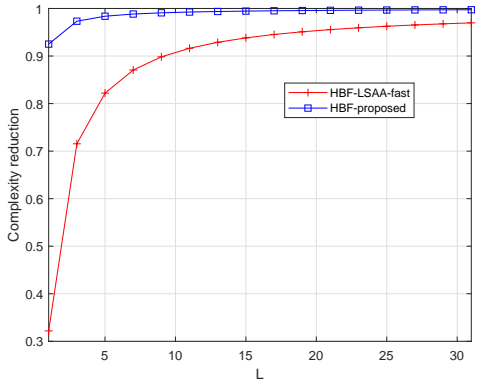}
         \captionsetup{font={small}}
        \caption{Computational reduction vs $L$ with $(N_r,Ns)=L*(8,1),N_r=N_t,N_s=N_{RF}$.}

        	\label{fig:Computational reduction}
        	\vspace{-0.3cm}
  \end{figure}

The results in Fig. \ref{fig:Computational complexity} show that the computational complexity of the three algorithms. The iteration number of HBF-LSAA-fast is $N_{iter}=10$ obtained by experiments. In the experiment, the number of transmitter and receiver antennas keeps the same, i.e., $N_r=N_t$, which is also the case for the data stream and RF chains, i.e., $N_{RF}=N_s$. The ratio between antennas and data streams is fixed. The network size is scaled by the coefficient $L$, which is the X-axis of Fig. \ref{fig:Computational complexity}. The initial system parameters are $N_r=N_t=8,N_s=N_{RF}=1$. The complexity reduction of the HBF-proposed and the HBF-LSAA-fast normalized by the HBF-LSAA is shown in Fig. \ref{fig:Computational reduction}. The results show that the proposed method has much lower computational complexity compared with the other two HBF algorithms. Specifically, the HBF-proposed can reduce over $99.7\%$ computational complexity compared with the HBF-LSAA at $L=31$. Moreover, the proposed method can achieve $99.08\%$ complexity reduction at $L=9$ while the HBF-LSAA-fast only achieves $89.83\%$ complexity reduction.

\section{Conclusion}
In this paper, we propose a low-complexity algorithm for HBF design in a mmWave MIMO-OFDM system. We obtain closed-form solutions to the analog and digital beamformers that maximize the upper bound of the SE. The numerical results show that the proposed method achieves state-of-art performance while requiring very low complexity. Therefore, the proposed scheme is feasible and practical for the implementation of mmWave massive MIMO systems. For future researches, the low-complexity HBF design in multi-cell and multi-user scenarios with imperfect CSI can be considered. 




\appendices
\section{Proof of Theorem \ref{tr:optimal solution}}
\label{sec_proof}
 The objective function in (\ref{eq:analog precoder problem}) can be written as
  \begin{align}\label{eq:equivalent RF beamformer without constraints upper bound}
   SE&=\log_2\left|\Ib_{N_{RF}}+\frac{\gamma}{\sigma_n^2}\Vb_{RF}^H\Hb_e\Vb_{RF}\right| \notag\\
        &=\sum_{i=1}^{N_{RF}}\log_2\left(1+\frac{\gamma}{\sigma_n^2}\lambda_{\Vb_{RF}^H\Hb_e\Vb_{RF},i} \right),
 \end{align}
 where $\lambda_{\Vb_{RF}^H\Hb_e\Vb_{RF},i}$ denotes the $i$th eigenvalue of $\Vb_{RF}^H\Hb_e\Vb_{RF}$. The matrix $\Vb_{RF}^H\Hb_e\Vb_{RF} \in \Cs^{N_{RF}\times N_{RF}}$ is a Hermitian matrix, and it is also the leading $N_{RF}\times N_{RF}$ principal submatrix of the $N_t \times N_t$ Hermintian matrix $\tilde{\Vb}^H\Hb_e\tilde{\Vb}$, where $\tilde{\Vb}$ is $N_t \times N_t$ unitary and expanded from $\Vb_{RF}$. Then, according to the interlacing property of the eigenvalues for Hermitian matrices, it can be shown that \cite{zhang2005variable}
 \begin{equation}\label{eq:interlacing property}
   \lambda_{\Vb_{RF}^H\Hb_e\Vb_{RF},i} \leq \lambda_{\tilde{\Vb}^H\Hb_e\tilde{\Vb},i}=\lambda_{\Hb_e,i}, 1\leq i \leq N_{RF}.
 \end{equation}
 Combined With \eqref{eq:equivalent RF beamformer without constraints upper bound} and  (\ref{eq:interlacing property}), we can prove that
 \begin{equation}\label{eq:upper bound by interlacing property}
   SE \leq \sum_{i=1}^{N_{RF}}\log_2\left(1+\frac{\gamma}{\sigma_n^2}\lambda_{\Hb_e,i} \right).
 \end{equation}
The equality holds when $\Vb_{RF}=[\vb_{e,1},\cdots,\vb_{e,N_{RF}}]$, where $\vb_{e,i}$ is the $i$th column of $\Vb_e$ which is given by $\Hb_e=\Vb_e\Lambdab_e\Vb_e^H$.

\bibliographystyle{IEEEtran}
\bibliography{conf_short,jour_short,PS_HBF}

\end{document}